\documentclass[letterpaper, 10 pt, conference]{ieeeconf}  

\IEEEoverridecommandlockouts                              

\usepackage{graphicx}

\usepackage{amsmath} 
\usepackage{amssymb}  
\usepackage{amsfonts}
\usepackage[ruled]{algorithm2e}
\usepackage{pifont}

\usepackage{booktabs}
\usepackage{hyperref}

\usepackage{pifont}

\newcommand{\cmark}{\ding{51}}%
\newcommand{\xmark}{\ding{55}}%

\newtheorem{theorem}{Theorem}[section]
\newtheorem{lemma}[theorem]{Lemma}

\newtheorem{cor}{Corollary}[section]

\newcommand{\E}{\mathcal{E}}
\newcommand{\N}{\mathcal{N}}
\newcommand{\OO}{\mathcal{O}}
\newcommand{\Nn}{\mathbb{N}}
\newcommand{\M}{\mathcal{M}}
\newcommand{\R}{\mathbb{R}}

\newcommand{\diag}[1]{\text{diag}(#1)}

\newcommand{\xd}{\dot{x}}
\newcommand{\yd}{\dot{y}}
\newcommand{\yb}{\bar{y}}
\newcommand{\pyx}{p(y_j \mid x_i)}

\newcommand{\Pyx}{P_{y \mid x}}
\newcommand{\Pxy}{P_{x \mid y}}

\newcommand{\Py}{P_y}
\newcommand{\Pby}{\bar{P}_y}
\newcommand{\Pbyx}{\bar{P}_{y \mid x}}
\newcommand{\Pbxy}{\bar{P}_{x \mid y}}

\newcommand{\SumiN}{\sum_{i\in\N}}
\newcommand{\SumjM}{\sum_{j\in\M}}
\newcommand{\SumijNM}{\sum_{\substack{i\in\N\\ j\in\M}}}
\newcommand{\Sumjjp}{\sum_{\substack{j\in\M\\ j'\in\M}}}
\newcommand{\Sumjp}{\sum_{j'\in\M}}
\newcommand{\Sumjneq}{\sum_{j' \neq j}}

\newcommand{\dxy}{d(x_i, y_j)}

\newcommand{\normf}[1]{\left\lVert#1\right\rVert_\infty}
\newcommand{\normm}[1]{\lVert#1\rVert^2}

\newcommand{\normxy}{\normm{x_i - y_j}}
\newcommand{\pa}[1]{\left(#1\right)}
\newcommand{\parr}[2]{\frac{\partial #1}{\partial #2}}

\title{\LARGE \bf
RCP: A Temporal Clustering Algorithm for Real-time Controller Placement in Mobile SDN Systems
}

\author{Reza Soleymanifar$^{1}$ and Carolyn Beck
\thanks{$^{1}$Authors are with Coordinated Science Laboratory, University of Illinois at Urbana-Champaign, United States of America
        {\tt\small reza2@illinois.edu, beck3@illinois.edu}}%
}

\begin{document}

\maketitle
\thispagestyle{empty}
\pagestyle{empty}

\begin{abstract}
Software Defined Networking (SDN) is a recent paradigm in telecommunication networks that disentangles data and control planes and brings more flexibility and efficiency to the network as a result. The Controller Placement (CP) problem in SDN, which is typically subject to specific optimality criteria, is one of the primary problems in the design of SDN systems. {\em Dynamic} Controller Placement (DCP) enables a placement solution that is adaptable to inherent variability in network components (traffic, locations, etc.). DCP has gained much attention in recent years, yet despite this, most solutions proposed in the literature cannot be implemented in real-time, which is a critical concern especially in UAV/drone based SDN networks where  mobility  is high and split second updates are necessary. As current conventional methods fail to be relevant to such scenarios, in this work we propose a real-time control placement (RCP) algorithm. Namely, we propose a temporal clustering algorithm that provides real-time solutions for DCP, based on a control theoretic framework for which we show the solution exponentially converges to a near-optimal placement of controller devices. RCP has linear $\OO(n)$ iteration computational complexity with respect to the underlying network size, $n$, i.e., the number of nodes, and also leverages the maximum entropy principle from information theory. This approach results in high quality solutions that are practically immune from getting stuck in poor local optima, a drawback that most works in the literature are susceptible to. We compare our work with a frame-by-frame approach and demonstrate its superiority, both in terms of speed and incurred cost, via simulations. According to our simulations RCP can be up to 25 times faster than the conventional frame-by-frame method.
\end{abstract}

\section{INTRODUCTION}

Software Defined Networking (SDN) has had major impacts on the performance of telecommunication networks, especially wireless networks, and it is predicted that this paradigm will play an integral role in 5G, UAV and IoT based networks \cite{Mouawad2018OptimalPlacement}.
The key idea in SDN is to dedicate a network component known as the \textit{controller-node} for centralization of network decision making in an attempt to disentangle the data plane (network packets) from the control plane (routing process). The immediate advantage of this architecture, among many others, is to allow the network administrator to reprogram the network without disrupting the data flow.

We define a Mobile SDN network as any SDN network with non-stationary network nodes. A wide range of such networks may be found in the literature, e.g., Software Defined Vehicular Networks (SDVNs), Low Earth Orbit (LEO) constellations, Software Defined Mobile Networks (SDMNs), Aeronautical Telecommunication Networks (ATNs) or drone/UAV networks. Similar to the work presented in \cite{Alharthi2019DynamicNetworks} and \cite{Sayeed2020EfficientNetworks}, we may assume controllers are SDN-enabled aerial devices enabling a flying network infrastructure, as illustrated in Figure 1. A practical realization of such devices could be drones, or UAV's. Deploying this type of SDN controller gives great flexibility in networks in locations where ground stations may not be present \cite{Singhal2019LB-UAVnet:SDN}.

The controller placement problem (CP) in SDN can be described as the task of finding the best location and node assignments for one or more controllers in the network, in the sense that a desired criterion is optimally achieved. In this setting, the function of the controllers is to disentangle the control and data flows. Latency between controllers and switches represents the main considered criterion by the majority of studies to date. Some additional metrics have been introduced including capacity and load balancing, inter-controller communication delay, deployment cost and energy consumption \cite{Toufga2020TowardsNetworks}. In recent years, high network node mobility and density, and the more wide-spread presence of wireless links with an overall higher degree of locational variability raise the issue of adapting the number of controllers and their locations according to these changes \cite{KalupahanaLiyanage2018ControllerNetworks}. As one example, LEO satellites spin around the Earth in periods less than 130 minutes leading to resource availability dynamics \cite{Papa2018DynamicNetwork}. According to the results of \cite{Luong2019Traffic-awareNetworks}, static optimal solutions do not adequately solve the CP problem in this setting, and it is necessary to develop solutions which account for the inherent dynamics present in the system. Following the work in \cite{Basta2015SDNPatterns}, we define dynamic controller placement (DCP) as an extension of CP that allows adapting to network dynamics. While CP has been around since Heller first proposed the problem in \cite{Heller2012TheProblem}, DCP is fairly new and is gaining increasing attention in the scientific community, due to the ever-growing need for more adaptive placement solutions.

\begin{figure} 
\label{fig:sdn}
\centering
\includegraphics[width = 8cm]{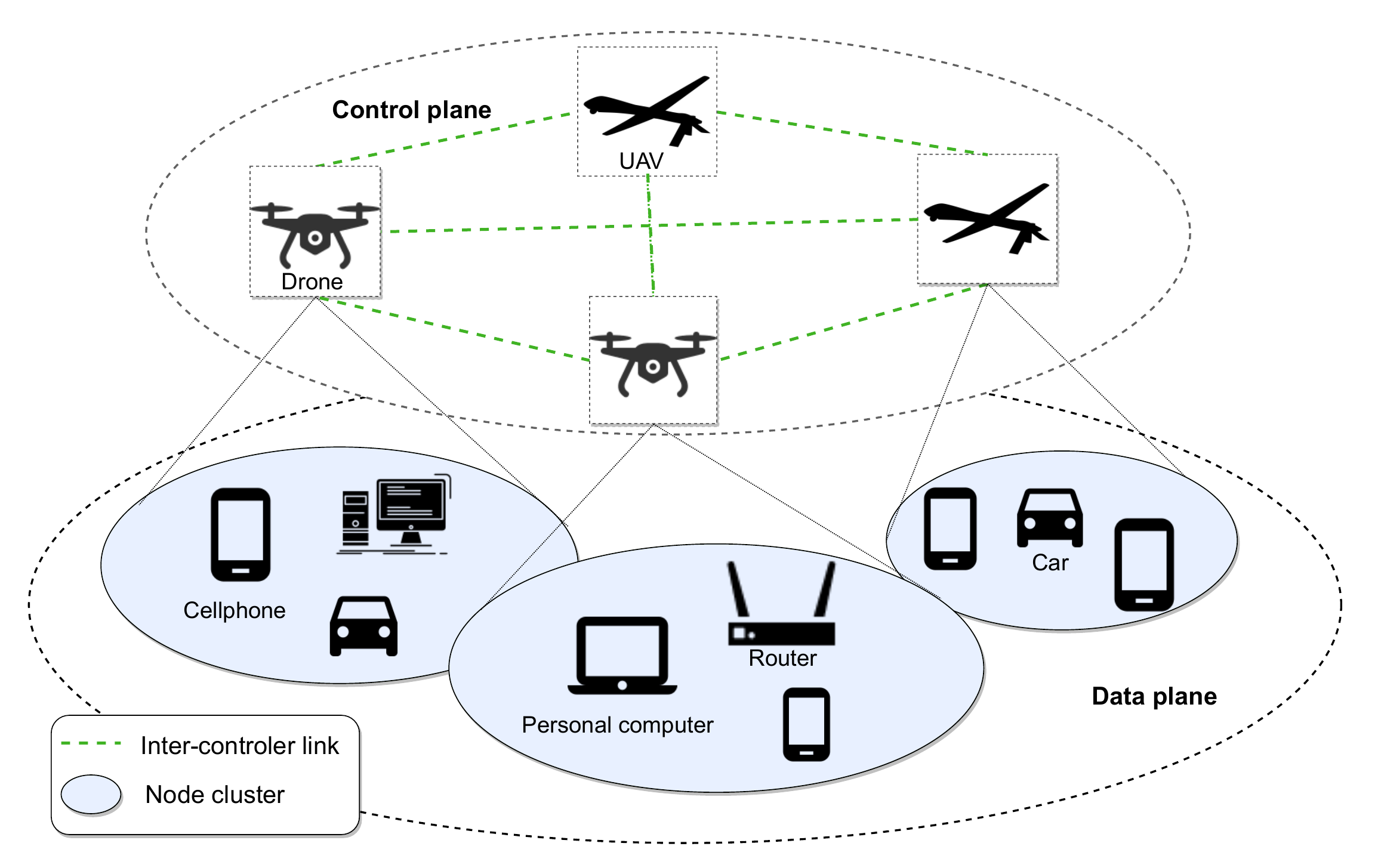}
\caption{A leaderless UAV/drone based SDN system with mobile network control nodes.}
\end{figure}

 We define a system snapshot, or frame, as a frozen state of the network at any fixed time. From an implementation standpoint, a dynamic controller localization problem can be regarded as a time-indexed set of static localization problems \cite{Sharma2012Entropy-basedProblems}. This perspective motivates a naive approach to DCP that is to re-solve CP every time a change in system occurs. However, this approach is clearly inefficient as previous solutions are not utilized and adapted to avoid unnecessary computations. We refer to such methods as \textit{frame-by-frame} approaches that ignore the temporal relationship between system snapshots and work with each frame in isolation. A pitfall of this approach is that if convergence time for CP is slower than the rate at which new network snapshots arrive--which is frequently the case for large-scale networks--this technique fails. To the best of our knowledge, all of the key works currently addressing the DCP problem in the literature fall in the frame-by-frame category.

Conventionally, mobility of the nodes is simply abstracted as a series of locations, thus rendering DCP a discrete optimization problem \cite{Toufga2020TowardsNetworks, Papa2018DynamicNetwork, Champagne2018ANetworking}. This type of modeling yields solutions that tend to be suboptimal compared to open search methods, but nevertheless makes the problem tractable. In this work we allow for a dynamic open search that can lead to better placements.
 
\begin{table*}[htbp!]
\label{tab:key_research}
\caption {Key research in DCP literature.}
\begin{tabular}{@{}lllllll@{}}
\toprule
Paper                                                                                             & Proactive & Context           & Objective                         & Variable             & Solution            & Real-time \\ \midrule
\cite{Toufga2020TowardsNetworks}                       & No        & SDVN              & Load balancing+Latency            & Traffic fluctuation  & ILP                 & \xmark        \\
\cite{Mouawad2018OptimalPlacement}                                                      & No        & 5G SDN            & Latency                           & Packet flows         & QP                  & \xmark        \\
\cite{Papa2018DynamicNetwork}                          & Yes       & LEO Constellation & Flow setup time                   & Network flows        & ILP                 & \xmark        \\
\cite{Alharthi2019DynamicNetworks}                                   & No        & Drone SDN         & Link quality                      & Node positions       & ILP                 & \xmark        \\
\cite{Champagne2018ANetworking}               & No        & SDN               & Load balancing+Latency            & Traffic fluctuation  & GA                  & \xmark        \\
\cite{ulHuque2017Large-ScalePlacement}                                                          & No        & SDN               & Maximum latency                   & Traffic fluctuation  & Heuristic           & \xmark        \\
\cite{He2019TowardApproach}            & Yes       & SDN               & Latency+migration                 & Traffic fluctuation  & SA                  & \xmark        \\
\cite{Luong2019Traffic-awareNetworks}
& No        & ATN               & Load balancing                    & Switch positions     & ILP+GA              & \xmark        \\
\cite{He2018HowPlane}                                                        & No        & SDN               & Flow setup time                   & Network flows        & ILP                 & \xmark        \\
\cite{Wu2018DynamicNetworking}                  & Yes       & SDSN              & Load balancing+Latency+economic   & Network flows        & APSO                & \xmark        \\
\cite{He2017ModelingFlows}            & No        & SDN               & Flow setup time                   & Network flows        & ILP                 & \xmark        \\
\cite{Wang2017AnNetworks}       & Yes       & Data Centers      & Latency+synchronization+migration & Request rate         & Game theoretic      & \xmark        \\
Our work                                                                                          & No           & Mobile SDN        & Latency+Synchronization           & Exact node positions & Temporal clustering & \cmark       \\ \bottomrule\bottomrule
\end{tabular}
\end{table*}
 
    Table 1 lists the approaches that have been used to solve DCP. For real-world networks with thousands to millions of nodes, which translates to the same order of constraints and variables in approaches like Integer Linear Programming (ILP) and Quadratic Programming (QP), the DCP problem cannot be solved in real-time, because the convergence times for these algorithms fall in the order of minutes to hours. To the best of our knowledge, this is the first work to solve the DCP problem in real-time, and does not rely on the frame-by-frame approach. Loosely speaking, RCP achieves this computational gain by removing the main loop, consisting of repetitive search in the solution space, which by default exists for any iterative optimization procedure that tries to find the placement of controllers. This comes at the cost of having sub-optimal placements at first, but exponentially converging to near-optimal placements as time goes by.

We classify the main contributions of our work as (1) finding high quality solutions for DCP in \textit{real-time}; (2) using the \textit{open search} method to explore all possible placements; (3) exploiting the temporal dynamics of the network for time complexity efficiency; (4) utilizing the maximum entropy principle in our algorithm to avoid poor local optima, resulting in better placements and assignments.

The rest of the paper is organized as follows. In section \ref{sec:lit-review} we review the recent research on controller-node placement in SDN. In section \ref{sec:prob-statement}, we introduce a mathematical framework for the controller-node placement problem in mobile SDN, and provide an in-depth statement of the problem. In Section \ref{sec:sol} we introduce our RCP algorithm and provide an analysis of its convergence properties.  Simulation results are given in Section \ref{sec:simulations} demonstrating the performance of our RCP algorithm, which we henceforth simply refer to as RCP. Finally we conclude the paper in Section \ref{sec:conclusion}, in which we summarize the contributions of our work and present possible future directions of this research.

\section{LITERATURE REVIEW}
\label{sec:lit-review}

In \cite{Papa2018DynamicNetwork}, in which a Low-Earth-Orbit (LEO) constellation setting is considered, the authors partition the time horizon into static snapshots and compute optimal placements and assignments of controllers at each time step using ILP. They assume the incoming flow at each snapshot is already known and {\em proactively} solve the DCP problem. They report that this ``dynamic" approach outperforms the standard static approach by approximately an average of 20\%.
In \cite{Alharthi2019DynamicNetworks}, the authors formulate the DCP problem for an SDN enabled drone network as a Mixed Integer Non-linear Program (MINLP), and propose solving this problem using a heuristic that relies on decoupling the placement and assignment tasks.

In \cite{Champagne2018ANetworking}, the authors use a multi-objective genetic algorithm (GA) to break SDN networks into domains and sub-graphs assigned to controllers. They use inter-controller latency, load distribution, and controller numbers as the fitness metrics of their GA algorithm.

In \cite{ulHuque2017Large-ScalePlacement}, the authors introduce the algorithm LiDy+, which has run-time complexity of $\OO(n^2)$ (an improvement over predecessors, with run time complexities of $\OO(n^2logn))$, and requires a smaller number of controllers while achieving a higher controller utilization. This method relies on heuristics for placing  controller modules and for adjusting the number of controllers needed according to traffic fluctuations.
The authors in \cite{Toufga2020TowardsNetworks} develop an ILP for a SDVN that updates the reallocation of roadside units (RSU) to their corresponding controllers. The dynamics of the network in this work is abstracted as the number of vehicles communicating with a RSU at each time step; here the ILP-based algorithm objective is to minimize a mixed latency and load balancing cost function.

The authors in \cite{He2019TowardApproach} cast DCP as a multi-period MINLP with partial information of future traffic flows. The authors consider both operational and migration costs and decompose the problem into smaller online problems, solving them using the Simulated Annealing (SA) algorithm. The authors in \cite{Luong2019Traffic-awareNetworks} design an ILP algorithm to address dynamic controller placement in Aeronautical Telecommunication Networks (ATNs). They propose two heuristic algorithms, DPFA and GA-DPDA, to solve the ILP problem when controller failure happens due to packet flow overload. In \cite{He2018HowPlane}, the authors introduce an ILP problem that considers both migration time and switch re-assignment time. Each time a new flow profile arrives this ILP is recalculated.

Software Defined Satellite Networking (SDSN) is considered in \cite{Wu2018DynamicNetworking} where authors partition the time horizon into smaller intervals within which the average flow per switch is assumed known. They further assume that back-up controller-nodes are placed throughout the network, and by toggling the on-off
status of controller-nodes they meet the changing network conditions. The solution approach proposed for this problem uses Accelerated Particle Swarm Optimization (APSO). In \cite{He2017ModelingFlows}, the authors report up to a 50\% improvement over static placement methods using a dynamic controller placement scheme that relies on solving an ILP that re-calculates optimal placements when system change occurs.

In \cite{Wang2017AnNetworks}, the authors consider the controller assignment aspect of DCP and decompose the problem into a series of stable matching problems with transfers, for which they propose a hierarchical two-phase algorithm that efficiently uses knowledge of future arrival rates. They report a 46\% reduction in cost and better load balancing compared to static assignment.

The authors in \cite{Basta2015SDNPatterns} model and analyze a realization of the mobile core network as virtualized software
instances running in data centers and SDN transport network elements, with respect to time-varying traffic demands. In \cite{Mouawad2018OptimalPlacement}, the authors develop a Quadratic Program (QP) that aims to minimize network switch to controller latency. As network packet flows change and controller overloads occur, they use another QP to perform switch migration to meet the increased load.

To summarize, there are two main approaches for the DCP problem in the literature. Either authors assume they proactively know the value of network variables in the future, as in \cite{Papa2018DynamicNetwork}, \cite{He2019TowardApproach}, \cite{Wu2018DynamicNetworking}, \cite{Wang2017AnNetworks}, or they rely on what we earlier described as the \textit{frame-by-frame} approach, as in \cite{Alharthi2019DynamicNetworks}, \cite{ulHuque2017Large-ScalePlacement} amongst others. Both of these methods fall short of real-world practicality. The former assumes availability of data that is typically not known, and the latter does not exploit the temporal relationship between network states over time. To overcome this gap, our work is aimed at creating a new placement procedure that can work in real-time and exploit given temporal relationships of the system; this has led to the design of the RCP algorithm.


\section{PROBLEM STATEMENT}
\label{sec:prob-statement}
We make the same assumptions on the communication protocol, controller type, and cost function as in  \cite{Qin2018SDNNetworks}. Specifically we assume network drones/UAV's are SDN-enabled and programmable via an API such as OpenFlow, similar to that described in \cite{Alharthi2019DynamicNetworks}. In this paper, we will maintain a high-level focus on the topology of the network and placement of controllers

For simplicity, our analyses are given in $\R^2$ equipped with the Euclidean norm; however our results can be extended directly to $\R^n$ (namely, to $\R^3$ in which the basic problem lies). We further assume that the domain of the problem $\Omega\subset \R^2$ is a compact set that serves as the space within which the network operates. Throughout the paper we use the shorthand notation $[A]_{ij} = a_{ij}$ to represent the matrix $A$ that is constructed by equating its $ij$th element to the scalar $a_{ij}$, and $A = \diag{a_j}$ to represent the diagonal matrix $A$ that has $a_j$ as the $j$th diagonal element. We also use $I_n$ to denote the identity matrix of size $n\times n$. We may occasionally drop the time index $t\in \R^+$ from the dynamical equations for the sake of readability and in such circumstances the reader can infer this from the context of the problem. We represent a mobile network as an undirected graph $G(\N,\E)$ with $\N$ as the set of nodes of the graph and $\E$ as the set of edges connecting these nodes. We represent the network nodes by the set $\N$ and the controller-nodes by the set $\M$. We will use $N =|\N|$ and  $M = |\M|$ to denote the number of regular network nodes and controller-nodes in the network, respectively. In the controller-node placement topology under study~\footnote{This topology is known as leaderless if there is no hierarchy amongst the controller-nodes.}, there is communication between controller-nodes and their assigned network nodes, and between controller-nodes themselves, however network nodes do not directly communicate with each other. In fact, each node is \textit{assigned} to a controller-node that serves as a gate between the node and the rest of the network \cite{Qin2018SDNNetworks}. In this work, we further assume that the delay in the network is proportional to the Euclidean distance between source and origin of connection. Let $x_i(t), y_j(t)\in \R^2$ represent the location of the network node $i\in\N$ and controller-node $j\in\M$, respectively, at time $t\in\R^+$. The dynamics of the network node $i\in \N$ is determined by the function $\phi_i(t, x, y): \R^+\times\R^{2N}\times\R^{2M}\rightarrow\R^2$ which we assume is continuously differentiable. The velocity of the controller-node $j\in\M$ is determined by the vector $u_j(t) \in \R^2$; it is this velocity function we aim to design to achieve real-time controller placement by solving a cluster tracking control problem. We thus represent a mobile SDN system by the following dynamical system:
\begin{equation}
\label{eq:system}
\begin{cases}
\xd =& \phi(t, x, y)\\
\yd =& u
\end{cases}
\end{equation}
where $x\in\R^{2N}$, $y\in\R^{2M}$, $\phi(t, x, y):\R^+\times\R^{2N}\times\R^{2M}\rightarrow\R^{2N}$, and $u\in\R^{2M}$ are vectors that are constructed by concatenation of network and controller node location and velocity vectors. Denoting $\zeta = [x^T y^T]^T\in \R^{2(N+M)}$ as the vector containing all positional information of the network nodes and controller-nodes, and letting $f(t, x, y) = [\phi(t, x, y)^T u(t)^T]^T\in \R^{2(N+M)}$ denote the vector containing the velocities of these nodes; then we compactly refer to this first order possibly nonlinear state-space system as
\begin{equation}
\label{eq:sys}
\dot{\zeta} = f(\zeta).
\end{equation}

Temporal clustering refers to separation of a time-indexed set of objects into disjoint partitions known as \textit{clusters} that satisfy a certain degree of similarity. The most representative point within each cluster is typically called the centroid, which in this work is equivalent to the location of the controller-node. Following the approach introduced by Rose (see \cite{Rose1998DeterministicProblems}), we leverage the Maximum Entropy Principle for our solution, letting $\pyx\in[0,1]$ denote the intensity of association of network node $i\in \N$ with controller-node (and cluster) $j\in\M$ such that $\SumjM \pyx = 1$; this quantity is also referred to as an association weight. We represent the objective function $F: \R^{N\times M} \times \R^{2N} \times \R^{2M}\rightarrow \R$ associated with this clustering problem by~\footnote{This is a relaxed version of the cost function in \cite{Qin2018SDNNetworks}. For details refer to \cite{Soleymanifar2020ABalancing}.}

\begin{align}
\label{eq:F}
    F(\Pyx, x, y) = 
    &\underbrace{
    \SumijNM \pyx \normm{x_i - y_j 
    }}_{D_1 :\text{delay cost}}
    \\\nonumber
    +\gamma&\underbrace{
    \Sumjjp \normm{y_j - y_{j'}} \SumiN \pyx
    }_{D_2 : \text{synchronization cost}}
    \\\nonumber
    -T&\underbrace{
    \pa{-\SumijNM \pyx \log \pyx
    }}_{\text{H: entropy}}
\end{align}
Here we aim to find a set of trajectories for $y_j$ which minimizes $F(\Pyx, x, y)$, thereby minimizing latency and synchronization times by using a maximum entropy function to help convexify the original problem. Note that so-called migration~\footnote{re-positioning controllers and re-assigning network nodes.} cost is outside the of scope of the present paper and so our objective function does not reflect this cost, as in \cite{Toufga2020TowardsNetworks}, \cite{Wu2018DynamicNetworking}, and \cite{Alharthi2019DynamicNetworks} among others.

\section{SOLUTION} \label{sec:sol}

 Based on the results of \cite{Soleymanifar2020ABalancing} and \cite{Xu2014ClusteringDynamics}, in which the authors extend the results of Rose to show that for a given set of trajectories, $\{y_j(t)\}$, a Gibbs distribution will minimize equation \eqref{eq:F}, we have 
\begin{equation}
\label{eq:probs}
\pyx = \exp{\left(-\frac{\dxy}{T}\right)}/Z_i\quad \forall i\in \N, j\in \M,
\end{equation}
where $\dxy = \normxy + \Sumjp\normm{y_j - y_{j'}}$ is the distance function, or so-called distortion measure~\footnote{Distortion is a term in information theory that signifies the dissimilarity or distance between two points.} between node $i$ and controller-node $j$, and $Z_i$ can be seen as the usual normalizing partition function. These association weights will become ``hard'' if we let $T$ go to zero, and uniform, as $T$ approaches a very high value. Formulation $\eqref{eq:F}$ and the ``annealing" or ``temperature" parameter $T$ ensures that the total system delay attains a good local minimum in theory, while the nodes are initially maximally noncommittal towards the controllers. This latter point is essential since according to the maximum entropy principle in information theory, among all candidate distributions, the one with highest entropy best describes the current state of the system. This approach also has advantages for optimization over the surface of the cost function \eqref{eq:F} where local optima abound \cite{Rose1998DeterministicProblems}. 

Using the terminology of \cite{Sharma2012Entropy-basedProblems}, let $[P_{y|x}]_{ij} = p(y_j | x_i) \in \R^{N \times M}$ be the matrix that contains information on the relative \textit{shape} of the clusters. Similarly define $[P_{x|y}]_{ij} = p(x_i | y_j) \in \R^{N \times M}$ as the matrix containing posterior associations (directly analogous to posterior probabilities) $p(x_i | y_j)$, which we calculate using Bayes' rule, with $p(x_i) = \frac{1}{N}$, for all $i \in \N$. Moreover, define $P_y =$ diag$(p(y_j)) \in \R^{M \times M}$ as the matrix containing information on the \textit{mass} of the clusters, where $p(y_j) = \sum_{i \in \N} p(y_j | x_i) p(x_i).$ Let $\Pbyx$,  $\Pbxy \in \R^{2N\times 2M}$ and $\Pby \in R^{2M \times 2M}$, such that $\Pbyx = \Pyx \otimes I_2$, $\Pbxy = \Pxy \otimes I_2$, and $\Pby = \Py \otimes I_2$. In prior work, we've shown that the optimal placement of controller $y$ with respect to energy function $F$ follows as:
\begin{equation}
    \label{eq:centroid}
    y = \Theta^{-1}\Pbxy^T x
\end{equation}
where $\Theta \in \R^{2M\times 2M}$ is a block matrix with $M^2$ blocks of size $2\times 2$. The diagonal blocks are equal to $\eta I_2$ where $\eta = \gamma(M-1) + 1$ and the non-diagonal blocks are equal to $-\gamma I_2$. We show that for $\gamma\neq \frac{1}{N-M}$ and $\gamma \neq \frac{1}{N - 2M}$, $\Theta^{-1}$ is well defined (see \cite{Soleymanifar2020ABalancing}). Ideally we want the function $F$ to serve as a control Lyapunov function, requiring the time derivative of $F$ along the trajectory of network nodes and controllers to be non-positive which then ensures a non-increasing value for \eqref{eq:F}, our objective function. Following the development in \cite{Soleymanifar2020ABalancing} and \cite{Xu2014ClusteringDynamics}, we can show the following.

\begin{theorem}
Given the control Lyapunov function \eqref{eq:F}, for the system defined in \eqref{eq:system} the time derivative of $F$ has the following structure.
\begin{equation}
\label{eq:derivative}
\dot{F} = 2\zeta^T\Gamma(\zeta)f(\zeta),\;
\Gamma(\zeta) = 
\begin{bmatrix}
&I_{2N\times 2N} &-\Pbyx\\
&-\Pbyx^T &N\Theta \Py
\end{bmatrix}
\end{equation}
\end{theorem}
\begin{proof}
Using basic calculus and taking partial derivatives of $F$ with respect to its constituents we can show that $\forall i\in \N$ and $j\in \M$,
\[
\parr{F}{x_i} = 2\pa{x_i \SumjM \pyx - \SumjM y_j \pyx}
\]
and 
\[
\parr{F}{y_j} = 2\eta N p(y_j)y_j - N\gamma p(y_j)\Sumjneq y_{j'}- \SumiN \pyx x_i.
\]
Thus so far we have shown that 
\[
\parr{F}{\zeta} = 2 \zeta^T \Gamma(\zeta).
\]
The next step is to apply the chain rule
\[
\frac{d F}{d t} = \parr{F}{\zeta}\parr{\zeta}{t} = \parr{F}{\zeta}f(\zeta),
\]
from which the desired result follows.
\end{proof}

Our goal here is to design a control law $u\in \R^{2M}$ such that the output of system \eqref{eq:system} asymptotically tracks the optimal placement of controllers in the 2D plane based on the control Lyapunov function \eqref{eq:F}. Following the results of \cite{R.SepulchreConstructiveBooks}, \cite{Sontag1983LYAPUNOV-LIKECONTROLLABILITY.}, \cite{Sontag1989AStabilization}, and \cite{Sharma2012Entropy-basedProblems} we propose the following control law, and show it results in a non-increasing function $F(t).$

\begin{theorem}
For the nonlinear system given in Equation \eqref{eq:sys} if 
\begin{equation}
\label{eq:control}
u = -\left[k_0 + \frac{(x^T - y^T\Pbyx^T)\phi}{\yb^T \Pby \yb}\right]\yb
\end{equation}
where $K_0 > 0$ is a positive scalar and $\yb = N \Theta\pa{y - \Theta^{-1}\Pbxy^T x}$, then $\dot{F}(t) \leq 0$ for all $t\in \R^+$.
\end{theorem}
\begin{proof}
We can expand equation \eqref{eq:derivative} to get
\[
\dot{F} = 2
\pa{
\pa{x^T - y^T\Pbyx^T}\phi +
\pa{-x^T\Pbyx + N y^T\Theta\Py}u
}.
\]
Now using the definition of $\yb$ and $\frac{1}{N}\Pbyx = \Pbxy \Pby$ we can show that 
\begin{equation} \label{eq:dotF}
\dot{F} = 2\pa{
\pa{x^T - y^T \Pbyx^T}\phi + \yb^T \Pby u
}.
\end{equation}
Substituting the control law \eqref{eq:control} into (\ref{eq:dotF}) gives us $\dot{F} = -2K_0\yb^T \Py \yb \leq 0$ for all $\yb \in \R^{2M}$, since $\Pby$ is assumed to be a positive definite matrix~\footnote{Note that degenerate (zero mass) clusters are not allowed in this formulation and diagonal elements of this diagonal matrix are always positive.} and $K_0 > 0$.
\end{proof}

To make explicit the result that $y$ asymptotically tracks the optimal placement of controllers we state Corollary \ref{cor:track}, which follows. We first introduce Lemma \ref{lem:bounded} without proof which is useful in the proof of Theorem \ref{cor:track}. For details please see Lemma E.1. in \cite{Sharma2012Entropy-basedProblems}.

\begin{lemma}
\label{lem:bounded}
For a non-negative function $f:\R \rightarrow \R$ of bounded variation, if $\int_0^\infty f(t)dt < \infty$ then $\lim_{t\rightarrow \infty} f(t) = 0$.
\end{lemma}

\begin{cor}
\label{cor:track}
For the dynamics given by system \eqref{eq:system} using the control law in \eqref{eq:control}, $y$ asymptotically tracks the optimal placement of controller nodes. That is:
\[
\lim_{t\rightarrow \infty} y(t) = \Theta^{-1}\Pxy x(t)
\]
\end{cor}

\begin{proof}
Because the system is constrained within the compact set $\Omega$, the continuous real-valued function $F$ is bounded from below. Since $\dot{F}\leq 0$ we must have that $F(t) \rightarrow F_\infty$ as $t\rightarrow \infty$. This implies that $\int_0^\infty -\dot{F}(\tau)d\tau = F(0) - F(\infty) < \infty$. Because $-\dot{F}$ is non-negative, using Lemma \ref{lem:bounded} we deduce that $\lim_{t\rightarrow\infty} -\dot{F}(t) = 0$. Since $-\dot{F}(t) = 2K_0\yb^T\Pby\yb$ and $\Pby$ is positive definite then we must have $\yb(t) \rightarrow 0$ as $t\rightarrow 0$. Using the definition of $\yb$ and because $\Theta$ is invertible by design, $y - \Theta^{-1}\Pbxy x\in \text{Null}(\Theta) = \{0\}$. This yields that $y(t) \rightarrow \Theta^{-1}\Pbxy x$ as $t\rightarrow \infty$ analogous to Equation \eqref{eq:centroid}.
\end{proof}

\begin{figure*}[!htbp]
\label{fig:trajectories}
\centering
\begin{tabular}{ccc}
\includegraphics[width=55mm]{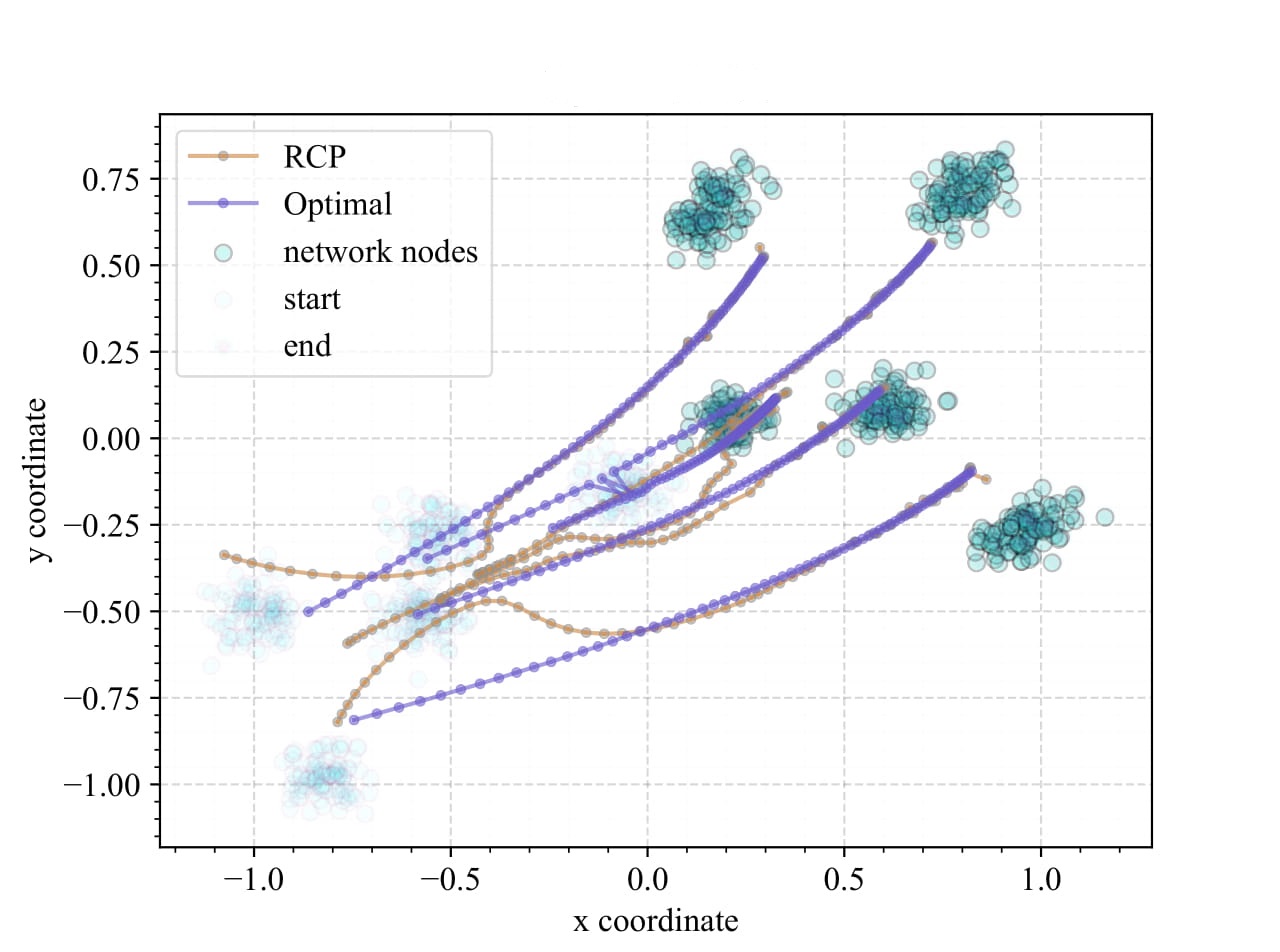}
 &\includegraphics[width=55mm]{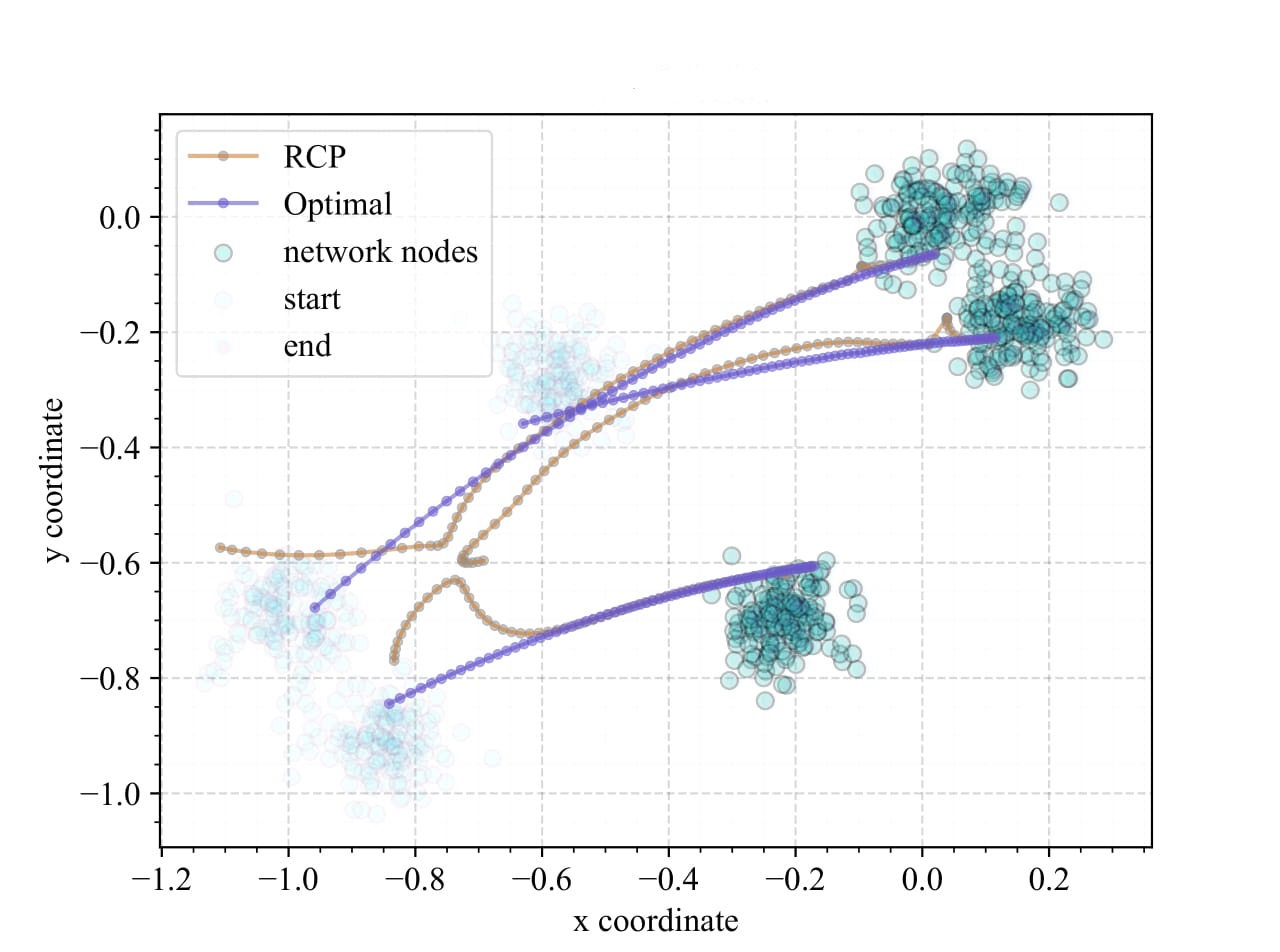}
 &\includegraphics[width=55mm]{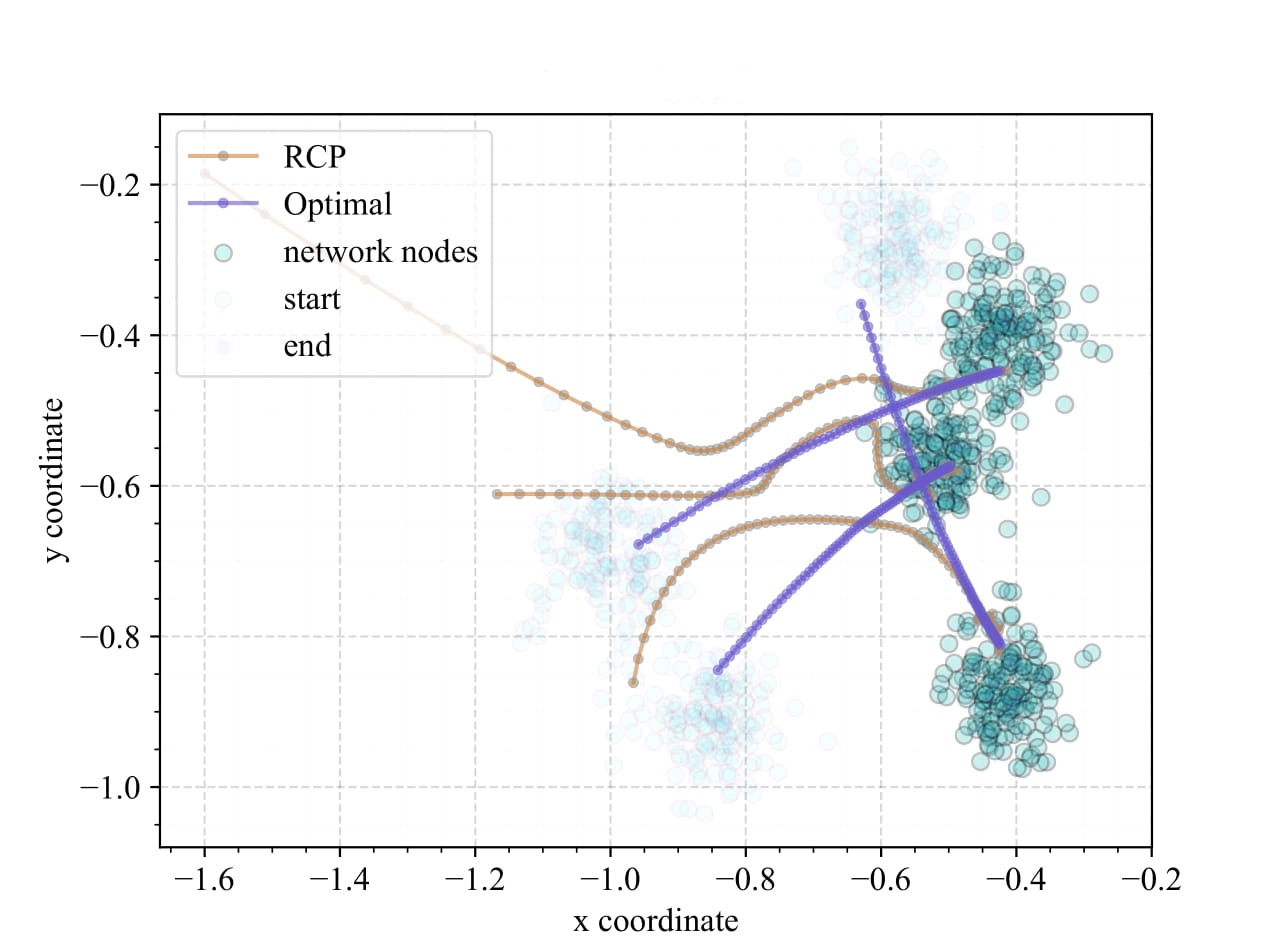}
\end{tabular}
\caption{Performance of RCP for four network topologies and mobility models.}
\end{figure*}

\begin{figure*}[!htbp]
\label{fig:sim}
\centering
\begin{tabular}{ccc}
\includegraphics[width=55mm]{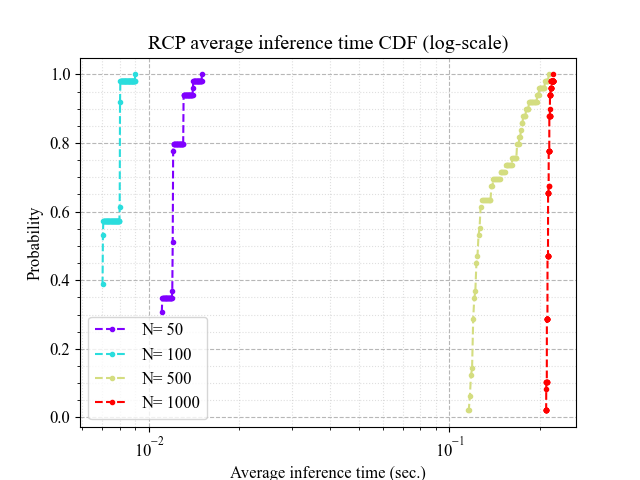}
 &\includegraphics[width=55mm]{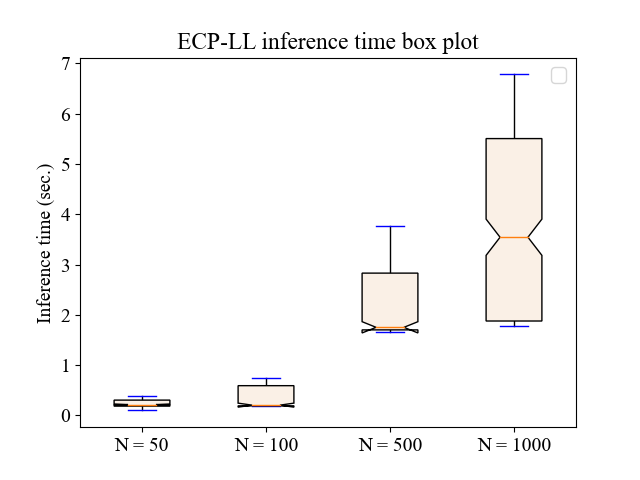}
&\includegraphics[width=55mm]{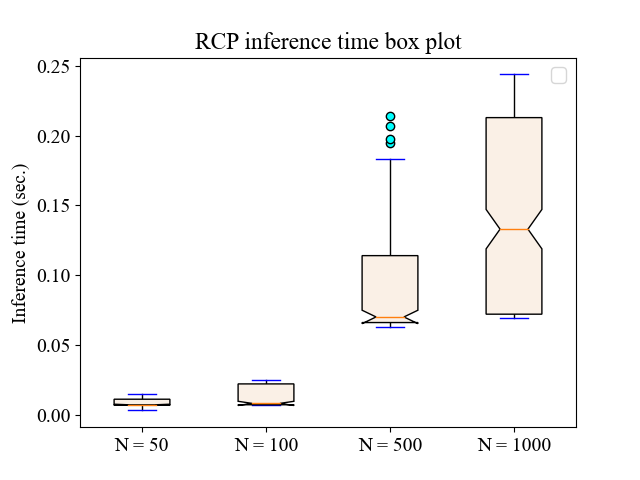}\\
(a) &(b) &(c)\\
 \includegraphics[width=55mm]{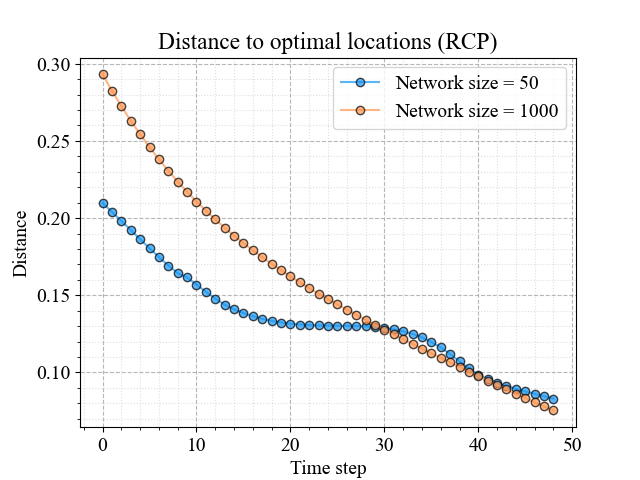}
 &\includegraphics[width=55mm]{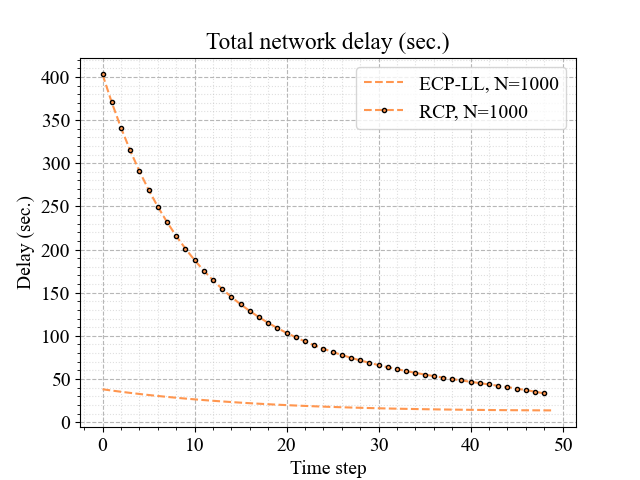}
 &\includegraphics[width=55mm]{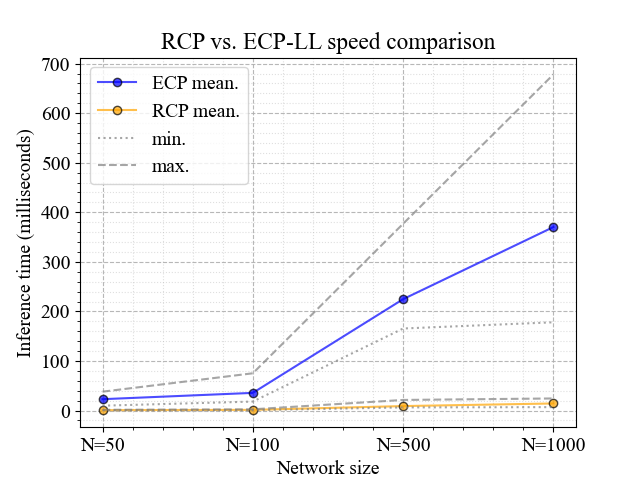}\\
 (d) &(e) &(f)
\end{tabular}
\caption{Performance of RCP for four network topologies and mobility models.}
\end{figure*}

\begin{algorithm}
\label{alg:main}
Initialize $y\in \R^{2M}$, starting temperature $T \approx \infty$, current time $t = 0$, time horizon $\tau\in \R^+$, time steps $n\in \Nn$, and decay rate $0<\alpha<1$\;

$\Delta t \leftarrow \frac{\tau}{n}$\;

\For{$i= 1$ \KwTo $n$}{
update time: $t \leftarrow t + (i - 1)\Delta t$\;

update association weights $\Pyx$ using \eqref{eq:probs}:
$
\small{
\pyx = \exp{\left(-\frac{\dxy}{T}\right)}/Z_i\quad \forall i\in \N, j\in \M}\;
$\;
update $y$ using control law (\ref{eq:control}):
\[
u = -\left[k_0 + \frac{(x^T - y^T\Pbyx^T)\phi}{\yb^T \Pby \yb}\right]\yb
\]
\[
y(t + \Delta t) \leftarrow y(t) + u(t)\Delta t
\]
update system temperature: $T \leftarrow \alpha T$
}
\caption{RCP}
\end{algorithm}

\section{Simulations}
\label{sec:simulations}


We performed simulations using Python 3.9.0 on a Razer Blade 15 laptop with Intel Core i7-10750H @ 2.60 GHz CPU, 16.0 GB RAM and Windows 10 Operating System. For the initial problem, data is generated as Gaussian distributions with randomized mean and standard deviations. Each starting cluster (distribution) is assigned to a destination cluster of equal size and each point in the initial cluster is assigned to a random point in the respective destination cluster. The network is modeled as a first-order linear dynamical system similar to \eqref{eq:system} with node positions given by $x(t): \R^+\rightarrow \R^2$, $x(t) = (x_{start} - x_{end}) \exp(-k t) + x_{end}$. At $t=0$ we start at $x_{start}\in \R^2$ and as $t\rightarrow \infty$ the system converges to $x_{end}\in \R^2$. 

The value $k$ for each point is randomized and is generated using the Rayleigh distribution with parameter $\sigma = 0.5$, and the probability distribution function $f(x; \sigma) = \frac{x}{\sigma^2}\exp(\frac{-x^2}{2\sigma^2})$. Rayleigh distributions are typically used for simulation of particle trajectories, which have coordinate-wise normally-distributed velocities. Due to arithmetic underflow that can occur in the floating point division in the update rule \eqref{eq:probs} we recommend the following normalization approach for any point $x \in \Omega$:
\begin{equation}
    x_{new} = \frac{x - \mu}{\sigma_{max} - \sigma_{min}}
\end{equation}
where $\mu = \frac{\int_{\Omega}  x dx}{\int_{\Omega} dx}$ is the mass center of the domain space $\Omega$. $\sigma_{min} = \min_{x\in \Omega} \normf{x}$ and $\sigma_{max} = \max_{x\in \Omega} \normf{x}$ are the minimal and maximal values the coordinates of points in $\Omega$ can take, with $\normf{.}$ being the infinity or max norm. After this normalization the network is restrained within the box $[-1, 1] \times [-1, 1]$

Figure 2 shows the trajectories generated by RCP algorithm. We encourage the reader to visit this link~\footnote{\url{https://www.youtube.com/playlist?list=PLkmxIANUXsFA2OCG7Uce1vBXYCqvvHSA3}} for animations of these controller placement simulations. Starting and ending positions of the clusters are slightly faded to indicate where the nodes have started and will end. The purple trajectories show the near-optimal~\footnote{It has been shown and is widely accepted that finding the globally optimal solution for controller placement is an NP hard problem. Here we work with a high quality locally optimal solution and refer to it as \textit{near-optimal} for convenience.} placement of controllers calculated using the frame by frame approach and the algorithm ECP-LL from \cite{Soleymanifar2020ABalancing}. This means that at each time step ECP-LL finds a near-optimal placement of controllers, shown by the purple dots along the optimal trajectories for a specific frozen snapshot of system. After random initialization within $\Omega$, RCP should ideally track these purple trajectories. 

Figure 3 shows various properties of RCP algorithm. Plot (a) shows the CDF~\footnote{Cumulative distribution function.} of inference time~\footnote{Here inference time is the time it takes to compute optimal placements for a system snapshot.} for RCP algorithm given various network sizes. Plot (b) and (c) show the quartiles of inference time for both RCP and ECP-LL algorithms using the box plot. Plot (d) shows the distance of RCP placements to the optimal controller locations. As can be seen this distance converges to zero as time goes by. Plot (e) shows the evolution of total network delay for both ECP-LL and RCP algorithms. Plot (f) compares the inference time of RCP versus ECP-LL algorithm across various network sizes. A more detailed statistics on the inference time can be seen in Table 2. These results show that RCP can be up to 25 times faster than conventional frame by frame approach using ECP-LL.


\begin{table}[htbp!]
\centering
\label{tab:stats}
\caption {Average inference time (Milliseconds).}
\begin{tabular}{lllll}
\hline
Network size & \multicolumn{2}{c}{ECP-LL} & \multicolumn{2}{c}{RCP} \\ \hline
             & average      & STD         & average      & STD      \\ \cline{2-5} 
50           & 22.92        & 8.38        & 0.83         & 0.30     \\
100          & 35.65        & 20.00       & 1.35         & 0.72     \\
500          & 224.79       & 67.29       & 9.03         & 3.22     \\
1000         & 370.30       & 167.19      & 14.27        & 6.56     \\ \bottomrule\bottomrule
\end{tabular}
\end{table}




\textbf{Computational Complexity} We can break down the computational complexity of RCP as follows: (1) calculating mutual distances between all $x_i, y_j$ pairs ($i\in\N$, $j\in \M$); (2) similar calculations for mutual distances between controllers; (3) calculating the distortions $d(x_i, y_j)$ for all $i\in \N$, $j\in \M$; (4) calculating association probabilities; and (5) updating the $y_j$ trajectories. The complexities for these operations are respectively: (1) $\OO(NMd)$, (2) $\OO(M^2d)$, (3) $\OO(Md)$, (4) $\OO(NM)$ and (5) $\OO(MNd^2)$. For terms (1) to (4) the calculations are analogous to those in the study completed in \cite{Soleymanifar2020ABalancing}. For term (5) the result comes from the fact that updating the $y_j$ requires calculating the numerator $\OO(MNd^2 + Nd)$ and denominator $\OO(M^2d^2)$, plus the final multiplication by $\yb$, $\OO(Md)$, which is dominated by the term $\OO(MNd^2)$ for $N>>M$. For a fixed time horizon $\tau$ we can express the overall computational complexity of RCP as $\OO(\tau NMd^2)$. This is a significant gain over other DCP algorithms like LiDy and LiDy+ discussed in \cite{ulHuque2017Large-ScalePlacement}, which respectively report $\OO(N^2)$ and $\OO(N^3)$ complexities in terms of network size.
Using the result in Corollary \ref{cor:track} it can be seen that in the limit RCP becomes a variation of the ECP-LL algorithm in \cite{Soleymanifar2020ABalancing} that is computed for only one iteration per each system snapshot. This observations roughly means that, in the limit, RCP is $\frac{T}{d}$ times faster than ECP-LL if on average ECP-LL takes $T$ iterations to converge for each system snapshot. This is because running ECP-LL over the time horizon will on average require $\OO(\tau TNMd)$ flops.

\section{CONCLUSION}
\label{sec:conclusion}
In this work we introduced RCP, a temporal clustering algorithm for real-time controller placement in mobile SDN systems. To the best of our knowledge, RCP is the first algorithm in DCP literature that exploits the temporal relationships of the network dynamics in order to efficiently adapt placement solutions in real-time. RCP leverages the principle of maximum entropy to avoid poor local optima that abound on the surface of our balanced cost function, and thus consistently provides high quality solutions. Unlike conventional methods that shrink the decision space into a discrete set, our algorithm allows use of the \textit{open search} method for placement, which makes it unlikely to yield sub-optimal solutions. RCP has linear $\OO(n)$ iteration computational complexity with respect to the network size and can be up to 25 faster than the conventional frame-by-frame approach. There are several analytical properties of the RCP algorithm such as sensitivity to temperature scheduling and convergence rate that we intend to address in follow-up work. A future direction for this work that can improve practicality of RCP is to include the estimation of the dynamics of the network, which we assumed herein is fully available to the decision maker. Another important direction is to incorporate a mechanism that enables RCP to dynamically change the number of controllers in the network.

\addtolength{\textheight}{-12cm}   



---







\bibliographystyle{ieeetr}
\bibliography{references.bib}

\end{document}